\documentclass[preprint]{imsart}

\RequirePackage{amsthm,amsmath,amsfonts,amssymb}
\RequirePackage[authoryear]{natbib}
\RequirePackage[colorlinks,citecolor=blue,urlcolor=blue]{hyperref}
\RequirePackage{graphicx}
\startlocaldefs
\newtheorem{example}{Example}
\newtheorem{theorem}{Theorem}
\newtheorem{definition}{Definition}
\newtheorem{lemma}{Lemma}
\newtheorem{corollary}{Corollary}

\endlocaldefs

\begin{document}
\begin{frontmatter}
\title{Minimally Discrete and Minimally Randomized p-Values}
\runtitle{p-Values for Discrete Data}
\author{Joshua D. Habiger and Pratyadipta Rudra} \\ \vspace{.1in} 
{\small Department of Statistics\\ Oklahoma State University}

\begin{abstract}
In meta analysis, multiple hypothesis testing and many other methods, p-values are utilized as inputs and assumed to be uniformly distributed over the unit interval under the null hypotheses.  If data used to generate p-values have discrete distributions then either natural, mid- or randomized p-values are typically utilized. Natural and mid-p-values can allow for valid, albeit conservative, downstream methods since under the null hypothesis they are dominated by uniform distributions in the stochastic and convex order, respectively.  Randomized p-values need not lead to conservative procedures since they permit a uniform distributions under the null hypotheses through the generation of independent auxiliary variates.  However, the auxiliary variates necessarily add variation to procedures.  This manuscript introduces and studies ``minimally discrete'' (MD) natural p-values,  MD mid-p-values and ``minimally randomized'' (MR) p-values. It is shown that MD p-values dominate their non-MD counterparts in the stochastic and convex order, and hence lead to less conservative, yet still valid, downstream methods.  Likewise, MR p-values dominate their non-MR counterparts in that they are still uniformly distributed under the null hypotheses, but the added variation attributable to the independently generated auxiliary variate is smaller. It is anticipated that results here will facilitate the construction of new meta-analysis and multiple testing methods via more efficient p-value construction, and facilitate theoretical study of existing and new methods by establishing gold standards for addressing the unavoidable detrimental ``discreteness effect''.   
\end{abstract}


\end{frontmatter}

\section{Introduction}
In statistical hypothesis testing the objective is to decide if a null hypothesis $H_0$, which specifies a distribution or statistical model for random data $X$, should be retained or rejected in favor of alternative hypothesis $H_1$ based on realization $X = x$. The decision can be denoted by $\delta(x)\in\{0,1\}$, where $\delta(x) =$ 0 (1) means that $H_0$ is retained (rejected).  The Neyman-Pearson (NP) paradigm stipulates that $\delta$ should  maximize power subject to the constraint that the type 1 error rate is not more than some prespecified level $\alpha$. 
Formally, $\delta$ should maximize $E_1[\delta(X)]$ subject to the constraint that $E_0[\delta(X)]\leq \alpha$, where $E_j$ denotes an expectation taken under $H_j$.   For example, in the simple versus simple setting  $H_0:X\sim P_0$ is tested against $H_1:X\sim P_1$, and the most powerful level $\alpha$ decision rule is $\delta(x) = I(\lambda(x)\geq c)$ where $I(\cdot)$ is the indicator function, $\lambda(x) = \frac{p_1(x)}{p_0(x)}$ is the likelihood ratio statistic, $p_j$ is the probability mass function for $P_j$ under $H_j$, and $c$ is chosen so that $E_0[\delta(X)] = \alpha$ \citep{Neyman1933}. The corresponding $p$-value can be computed via $p(x) = \Pr_0\{\lambda(X)\geq \lambda(x)\}$.  

However, if $X$ has a discrete distribution, it is not generally possible to define $\delta(X)$ such that $E_0[\delta(X)]=\alpha$, and several different types of $p$-values can be considered.  Let us consider a simple example to facilitate discussion.  
\begin{example}\label{BernoulliExample}
Let $X = (X_1, X_2, ..., X_5)$ be iid Bernoulli($\theta$) random variables with joint probability mass function $p_\theta(x) = \theta^{T(\mathbf{x})}(1-\theta)^{5-T(\mathbf{x})}$ where $T(x) = \sum_ix_i$, and consider testing null hypothesis $H_0:\theta = 0.5$ against alternative hypothesis $H_1:\theta = 0.8$.
\end{example}
\noindent Observe that $\lambda(x) = p_{1}(x)/p_{0}(x)$ is increasing in $T(x)$ so the likelihood ratio test rejects $H_0$ for large values of $T(x)$. However, observe $T(x)$ takes values in $\{0,1,...,5\}$ and that $\Pr_0\{T(X) \geq 5\} = 0.03125$ but $\Pr_0\{T(X)\geq 4\} = 0.1875$.  Hence, it is not possible to define a $\delta$, i.e. choose $c$, such that $E_0[\delta(X)] = \Pr_0\{T(X)\geq c\} = \alpha$ if, for example, $\alpha = 0.1$.  Randomized testing \citep{Pearson1950, Tocher1950} would be necessary. In Example \ref{BernoulliExample}, a most powerful level $0.1$ test function is 
\begin{equation}\label{BernoulliTestFunction}
\phi(x) = \left\{\begin{array}{rcl} 1 & \mbox{if} &T(x) > 4 \\
 0.44 & \mbox{if} & T(x) = 4  \\
 0 & \mbox{if}& T(x) < 4.
 \end{array}\right.
\end{equation}
where $\phi(x)$ represents the rejection probability. See, for example, \cite{Lehmann1993}. The randomized test, when implemented, can be written as a randomized decision rule $\delta(x,u) = I(u\leq \phi(x))$ where $u$ is a realization from a standard uniform distribution. The corresponding randomized $p$-value can be computed $p(x,u) =  \Pr_{0}\{T(X)>T(x)\}+ u\Pr_{0}\{T(X) = T(x)\}$.  See, for example, \cite{Pena2011, Habiger2011}.   In practice, we may also report the corresponding natural $p$-value: $p(x,1)=\Pr_{0}\{T(X)\geq T(x)\}$ (cf. \cite{Casella2002}), mid-p-value \citep{Lan61}: $p(x,1/2) = \Pr_{0}\{T(X)>T(x)\}+ 1/2\Pr_{0}\{T(X) = T(x)\}$, or even the fuzzy $p$-value \citep{GeyMee05}, which in this example amounts to reporting the interval $[p(x,0), p(x,1)]$.  The curious reader is referred to \cite{Agr07, Wells2010} for comparison of approaches in traditional single null hypothesis testing settings. More recent research studies the impact of the employing mid-p-values, natural p-values, or randomized in downstream methods. 


First let us consider randomized $p$-value literature. \cite{Habiger2011, Dickhaus2013} provided conditions that ensure randomized $p$-values are uniformly distributed under the null hypotheses so that the \cite{Benjamini1995} multiple hypothesis testing procedure, and other adaptive procedures (eg. \cite{Storey2004}) for false discovery rate (FDR) control are valid.  \cite{Ochieng2024} combined randomization methods in \cite{Hoang2022,Hoang2022b}  with aforementioned randomization methods to develop improved 2 step procedures for multiple testing of composite null hypotheses with discrete data. \cite{Habiger2015, Dai2019} demonstrated that randomized p-values can be used to develop more powerful multiple testing procedures even if (adjusted) randomized $p$-values are not ultimately reported, say as in the fuzzy BH procedure in \cite{Kulinskaya2009}. Of course, randomized $p$-values may be more variable than their non-randomized counterparts, which can be manifested as decreasing replication probability. That is, the probability of realizing the same randomized p-value in a replicated experiment is 0 if $u$ is regenerated. 

Other research is aimed at establishing conservative (or less conservative) behavior for mid and natural p-values.  For example, \cite{Doehler2018, Chen2020, Chen2020b} provide upper bounds for the FDR when the BH procedure is applied to mid-p-values.   In meta analysis, \cite{Rubin-Delanchy2018} provide bounds for the distribution of Fisher's test statistic $T= -2 \sum_i \log(P_i)$ under the null hypothesis when $P_i$ represents a mid-p-value. The idea builds upon \cite{Hwang2001}, where it was shown that the mid-p-value is stochastically less than or equal to a uniform distribution in the convex order \citep{Shaked2007} under the null hypothesis.  \cite{Wang2024} also used stochastic ordering to calibrate mid-p-values into valid p-values, valid E-values \citep{Vovk2021} and Bayes Factors. Of course, established inequalities need not be sharp. 

Ideally inequalities for downstream methods that utilize mid-p-values or natural p-values would be as sharp as possible and, if opting for randomized $p$-values, generated uniform variates would add minimal variation.  This motivates our minimality principle, which ensures that the aforementioned detrimental impacts when dealing with the ``discreteness'' of the data are minimized.  Specifically, we develop \textit{minimally discrete} (MD) mid-p-values,  MD natural $p$-values, and \textit{minimally randomized} (MR) p-values that dominate their traditional mid-$p$-value, natural $p$-value and randomized $p$-value counterparts, respectively. Specifically, it is shown that inequalities based on the usual stochastic order for natural $p$-values and inequalities based on the convex order for mid-$p$-values are sharpened.  Likewise, variation attributable to independeintly generated auxiliary variates for randomized $p$-values is smaller.  

Our implication that ``less discrete'' and ``less randomized'' p-values are available is bold and therefore warrants a brief discussion. Reconsider Example \ref{BernoulliExample} and let us compare \begin{equation}\label{AnotherMPTest}
\phi^*(x) = \left\{\begin{array}{rcl} 1 & \mbox{if} & x\in \{ (1,1,1,1,1), (1,1,1,1,0),  (0,1,1,1,1)\}  \\ 
 0.2 & \mbox{if} & x=(1,1,0,1,1)  \\ 
 0 & & \textrm{otherwise}
 \end{array}\right.
\end{equation}
to the most powerful test $\phi(x)$ in \eqref{BernoulliTestFunction}. A curious reader may verify that the type 1 error rates are $E_0[\phi(X)] = E_0[\phi^*(X)] = 0.1$ and that the powers are $E_1[\phi(X)] = E_1[\phi^*(X)] = 0.5079$. Clearly $\phi^*$ and $\phi$ are two distinct most powerful level $0.1$ tests. In particular, both tests automatically reject $H_0$ when $T(x)>4$ and automatically retain $H_0$ when $T(x)<4$.  The difference is that $\phi(x)$ rejects $H_0$ with probability 0.44 for all five $x\in [T(x)=4]$ while $\phi^*(x)$ only randomly rejects $H_0$ for one $x\in [T(x)=4]$.  It automatically rejects $H_0$ if $x = (1,1,1,1,0)$ or $x= (0,1,1,1,1)$ and automatically retains $H_0$ when $x=(1,0,1,1,1)$ or $x = (1,1,1,0,1)$.  The key realization is that a most powerful test function can be defined ``arbitrarily on the set $[\lambda(x) = k]$'' \citep{Lehmann1997} as long as the resulting test function has expectation equal to $\alpha$ under $H_0$.  While mathematical statistics texts (cf. \cite{Cox1974, Lehmann1997}) acknowledge this point, further development is cited as beyond scope: in testing a single null hypothesis $\phi$ is sometimes viewed as a mathematical mechanism for illustrating the NP Lemma rather than a practical method.  However, as noted previously, p-values derived from test functions like $\phi$ are increasingly utilized as inputs in a downstream method.  In this setting, further research is warranted.        


\section{Setup}
We first provide a basic mathematical framework for test functions, decision functions and $p$-values and  recall / establish some fundamental results for mid-$p$-values, natural $p$-values and randomized p-values.  Let $X = (X_1, X_2, ..., X_n)\in \mathcal{X}$ be a collection of independent and identically distributed random variables with distribution $\Pr_\theta$ for $\theta\in \Theta$.  Assume that the support $\mathcal{X}$ is countable so that $\Pr_\theta$ is discrete.  Denote the probability mass function of $X$ under $\Pr_\theta$ by $p_\theta(x) = \Pr_\theta\{X = x\}$ and an expectation taken under $X\sim \Pr_\theta$ by $E_\theta$. We shall focus on testing a simple null hypothesis $H_0:\theta = \theta_0$ against an alternative hypothesis, say $H_1:\theta=\theta_1$ for $\theta_1\neq \theta_0$ or $H_1:\theta\in\Theta\setminus{\theta_0}$. For ease of exposition, we sometimes denote $E_{\theta_0}$ by $E_0$ and write $E_1$ for an expectation taken under $H_1$.  We shall adopt similar notations for probabilities and probability mass functions under $H_0$ or $H_1$. 

A size $\alpha$ test function is denoted $\phi_\alpha: \mathcal{X}\rightarrow [0,1]$ and satisfies $E_0[\phi_\alpha(X)] = \alpha$. Here, the size $\alpha$ also indexes a test function in the collection $\{\phi_\alpha(X), \alpha\in [0,1]\}$.  Assume $\phi_\alpha(x)$ is nondecreasing and right continuous in $\alpha$ for every $x$ and that $\phi_\alpha$ is unbiased in the that $E_\theta[\phi_\alpha(X)]\geq \alpha$ for all $\alpha\in[0,1]$ and $\theta\in \Theta$.  In practice, for a specified $\alpha$, $\phi_\alpha(x)$ takes on values 0 or 1 for most realizations of $x\in \mathcal{X}$, in which case $H_0$ is automatically retained or rejected, respectively.  However, it is necessary to allow $\phi_\alpha(x)\in (0,1)$ for some $x$ to ensure that $E_0[\phi_\alpha(X)]=\alpha$, as in Example \ref{BernoulliTestFunction}.  If $\phi_\alpha(x)\in(0,1)$ and a decision to reject or retain $H_0$ is mandated or a $p$-value is to be computed, then further notation/action is required.

A \textit{randomized decision function} corresponding to $\phi_\alpha$ is defined as $\delta_\alpha(X,U) = I(U\leq \phi_\alpha(X))$, where $U$ is an independently generated uniform random variable over the unit interval, and where $\delta = 0~(1)$ means that $H_0$ is retained (rejected).  Observe that if $\phi_\alpha(x) = 0~(1)$, then $H_0$ is retained (rejected) with probability 1, but that if $\phi_\alpha(x)$ is in $(0,1)$ then $H_0$ is rejected with probability $\phi_\alpha(x)$ since $\Pr(U\leq \phi_\alpha(x)) = \phi_\alpha(x).$ Observe that
\begin{equation}\label{size alpha}
E_0[\delta_\alpha(X,U)] = E_0\left\{E_0[I(U\leq \phi_\alpha(X))|X]\right\} = E_0[\phi_\alpha(X)] = \alpha,
\end{equation} where we adopt shorthand notation $E[\cdot|X]$ for $E[\cdot|\sigma(X)]$. We therefore refer to $\delta_\alpha(X,U)$ as the size $\alpha$ decision function corresponding to $\phi_\alpha(X)$.  
A \textit{non-randomized decision function} fixes a $u\in [0,1]$, say $u = 1/2$ or $u = 1$, and is denoted $\delta_\alpha(X,u)$.  Taking $u = 1$ and observing that $\delta_\alpha(x,1) = I(1\leq \phi_\alpha(x)) \leq \phi_\alpha(x)$, we have
\begin{equation}\label{level alpha}
E_0[\delta_\alpha(X,1)] \leq E_0[\phi_\alpha(X)] = \alpha.
\end{equation}
Due to the above inequality, $\delta_\alpha(X,1)$ is referred to as a level $\alpha$ decision function corresponding to $\phi_\alpha$.  If $u<1$ is specified, then $\delta_\alpha(X,u)$ is still non-randomized but may or may not be a level $\alpha$ decision function because the inequality in \eqref{level alpha} need not be satisfied.  See, for example, Theorem 2 in \cite{Habiger2015}.

Above we see that $u$ can be specified or generated, thereby rendering $\delta$ non-randomized or randomized, respectively.  Likewise, $p$-values corresponding to $\delta$ can be randomized or non-randomized. This leads us to define the \textit{generalized} $p$-value for $\delta_\alpha(x,u)$ by
\begin{equation}\label{p-value}
P(x,u) = \inf\{\alpha:\delta_\alpha(x,u) = 1\}.
\end{equation}
When $U$ represents a random variable, $P(X,U)$ is referred to as a randomized $p$-value statistic.  If $u$ is fixed or specified then $P(X,u)$ is a non-randomized $p$-value statistic.  The generalized $p$-value can be interpreted in the usual way: $P(x,u)$ is the smallest $\alpha$ allowing for $H_0$ to be rejected when $X=x$ and $U=u$.  See, for example, \cite{Casella2002}, pg. 412.  We shall sometimes refer to $P(X,U)$ as a p-value rather than a $p$-value statistic for brevity.  

This $p$-value is well defined in the sense that $\delta_\alpha(X,U) = I(P(X,U)\leq \alpha)$ with probability 1.  See, for example, Theorem 2.3 in \cite{Habiger2011}. In fact $\delta_\alpha(X,u) = I(P(X,u)\leq \alpha)$ with probability 1 for any fixed $u\in [0,1]$, including $u=1/2$ or $u = 1$.  See Theorem 1 in \cite{Habiger2015}.  Consequently,
\begin{equation}\label{valid p-value}
\textrm{Pr}_0\{P(X,U)\leq \alpha\} = \alpha \mbox{ and } \textrm{Pr}_0\{P(X,1)\leq \alpha\} = E_0[\delta_\alpha(X,1)] \leq \alpha.
\end{equation}
In summary, the \textit{randomized} $p$-value statistic $P(X,U)$ is uniformly distributed under $H_0$ while the \textit{natural} $p$-value statistic $P(X,1)$ is stochastically greater than or equal to a uniform random variate under $H_0$. It is worth reiterating that all equalities and inequalities in \eqref{valid p-value} are true under our working assumptions: $\phi_\alpha(x)$ is unbiased, $\phi_\alpha(x)$ is non-decreasing and right-continuous in $\alpha$, and $p(x,u)$ is computed via \eqref{p-value}. An additional condition that allows us to compare mid-p-values to a uniform distribution is formalized next.

While we still have $I(P(X,1/2)\leq \alpha) = \delta_\alpha(X,1/2)$ with probability 1, we cannot provide bounds for $P(X,1/2)$ in the usual stochastic order under $H_0$ since $E_0[\delta_\alpha(X,1/2)]$ can be less than or greater than $\alpha$.  However, \cite{Hwang2001} showed that $E_0[h(P(X,1/2))]\leq E[h(U)]$ for any convex function $h$ for $p$-values arising from the analysis of $2\times2$ contingency tables.  That is $P(X,1/2)$ is stochastically less than or equal to $U$ in the convex order \citep{Shaked2007}.  Note that sometimes this stochastic ordering is written $P(X,1/2) \leq_{cx} U$ under $H_0$ and $P(X,1/2)$ is referred to as \textit{subuniform} under $H_0$. The convex ordering result is easily generalized beyond \cite{Rubin-Delanchy2018, Hwang2001} to any $p$-value that is linear in $u$, i.e. any p-value that can be written $p(x,u) = a(x) + u b(x)$ since this allows for $E_0[p(x,U)] = p(x, E_0[U]) = p(x,1/2)$. This is formally stated below. 
\begin{lemma}\label{sub-uniform-lemma}
Let $P(x,u)$ be a generalized $p$-value defined as in \eqref{p-value} with $\Pr_0\{P(X,U)\leq t\}=t$ for $t\in[0,1]$, and with $P(X,U) \stackrel{d}{=} a(X) + U b(X)$ for some $a(X)$ and $b(X)$.  Then $P(X,1/2)$ is sub-uniform under $H$.   
\end{lemma}
\begin{proof}
We establish that
\begin{eqnarray*}\label{sub-uniform}
E[h(U)]&=& E_0 [h(P(X,U))]\\
&=&E_0[E_0\{h(P(X,U))|X\}]\\ 
&\geq& E_0[h(E_0\{P(X,U)|X\})] \\
&=& E_0[h(P(X,1/2))]
\end{eqnarray*}
for $h$ convex so that the result follows by the definition of sub-uniform. To see this, observe the first equality is by the supposition that $P(X,U) \stackrel{d}{=}U$ under $H_0$.  The second equality utilizes the law of iterated expectation.  The inequality is due to Jensen's inequality.  The last equality follows from the linearity condition and that $E_0[U]=1/2$.  In particular, using linearity and noting that $X$ and $U$ are independent gives 
$$E_0\{P(X,U)|X\} = E_0\{a(X) + U b(X)|X\} = a(X) + E_0[U] b(X) = P(X,1/2).$$
\end{proof}
\noindent Above we see that $p(x,1/2) = p(x,E_0[U])$.  This establishes that the mid-p-value can be viewed as a smoothed version of the randomized $p$-value, which leads to the terminology ``expected p-value'' in \cite{Hwang2001}.        

The linearity condition is satisfied whenever test functions depend on a statistic $T(X)$, where large (or small WLOG) values of $T(x)$ are evidence against $H_0$.  This arises in most practical settings, such as when $T(X)$ is a likelihood ratio test statistic or some statistic that is monotone in the likelihood ratio. The corresponding size $\alpha$ test function, decision function, and generalized $p$-value are formally defined below.
\begin{definition}[Tests and p-values for $T$]\label{T}
Let $T:\mathcal{X}\rightarrow \Re$ be a test statistic such that large values of $T$ are evidence against $H$.  A size $\alpha$ test function based on $T$ is
$$\phi_\alpha^{T}(x) = \left\{\begin{array}{rcl} 1 & \mbox{if} &T(x) > k(\alpha) \\
 \gamma(\alpha) & \mbox{if} & T(x) = k(\alpha)  \\
 0 & \mbox{if}& T(x) < k(\alpha).
 \end{array}\right.
$$
where $\gamma(\alpha)$ and $k(\alpha)$ are chosen such that $E_0[\phi_\alpha^T(X)] = \alpha$. Its corresponding decision function is $\delta_\alpha^T(x,u) = I(u\leq \phi_\alpha^T(x))$ and generalized $p$-value is
$
P^{T}(x,u) = \mathrm{Pr}_{0}\{T(X)> T(x)\} + u \mathrm{Pr}_0\{T(X) = T(x)\}.
$
\end{definition}
Combining Lemma \ref{sub-uniform-lemma} and Definition \ref{T}, we define a large class of mid-p-value statistics that are sub-uniform under $H_0$ in the following corollary.  It is worth noting that Definition \ref{T} explicitly outlines how to construct a mid-p-value that is sub-uniform under $H_0$ as long as some test statistic $T$ is available.  This includes, but is not limited to, likelihood ratio statistics, sufficient statistics for some parametric models, and nonparametric rank-based statistics.  
\begin{corollary}
Let $P^T(x,1/2)$ be a mid-p-value defined as in Definition \ref{T}.  Then $P^T(X,1/2)$ is sub-uniform under $H_0$.
\end{corollary}

\section{Minimality, Minimally Discrete p-Values and Minimally Randomized p-Values}
Observe that because the support $\mathcal{X}$ of $X$ is countable, each individual point in $\mathcal{X}$ can be ranked according to the order in which it is to be included in a rejection region.  In Example \ref{BernoulliExample} the point $x = (1,1,1,1,1)$ may be ranked first, $x = (0,1,1,1,1)$ ranked second, $x = (1,1,1,1,0)$ third and so on until each of the $x\in\mathcal{X}$ have been ranked from $1$ to $32$. Here, it is important to understand that these rankings, say $R(x)$ \textit{agree} with rankings provided by the test statistic $T(x)$ in \eqref{BernoulliTestFunction} from the most powerful LR test.  That is, while neither $\lambda(x)$ nor $T(x)$ require that the point $x = (0,1,1,1,1)$ be ranked second, they allow for it. It is precisely this distinction that is studied in this manuscript.  

First let us define minimally discrete test functions and their corresponding decision functions and p-values.  Formally, a test function that depends on $X$ through some one-to-one mapping of $\mathcal{X}$ onto the positive integers (WLOG), say $R:\mathcal{X}\rightarrow \mathcal{N}$, is referred to as minimally discrete (MD).  When non-randomized decision functions and $p$-values are constructed using an MD test function, they are also referred to as MD.  Randomized decision functions and randomized $p$-values based on MD test functions are said to be minimally randomized (MR) rather than minimally discrete since, for example, the resulting $p$-value has a continuous distribution.  Decision functions and generalized $p$-values that do not specify whether $u$ is fixed or generated are referred to as M-p-values and M-decision functions whenever they are based on MD test functions. Formal definitions are below.      
\begin{definition}[MD tests, decision functions, and p-values]
\label{MD}
Let $R:\mathcal{X}\rightarrow \mathcal{N}$ be one to one, where $\mathcal{N}$ represents the positive integers, and define
$$
\phi_\alpha^{MD}(x) = \left\{\begin{array}{rcl} 1 & \mbox{if} &R(x) <k^*(\alpha) \\
 \gamma^*(\alpha) & \mbox{if} & R(x) = k^*(\alpha)  \\
 0 & \mbox{if}& R(x) > k^*{\alpha}.
 \end{array}\right.
$$
where $\gamma^*(\alpha)\in[0,1]$ and $k^*(\alpha)\in \mathcal{N}$ are such that $E_0[\phi_\alpha^{M}(X)] = \alpha$.  Then $\phi_\alpha^{MD}(x)$ is called a minimally discrete size $\alpha$ test function.  Its corresponding $M$-decision function is $\delta_\alpha^{M}(x,u) = I(u\leq \phi_\alpha^{MD}(x))$ and $M$-$p$-value is
$
P^{M}(x,u)  = \mathrm{Pr}_0\{R(X)< R(x)\} + u\mathrm{Pr}_0\{R(X)=R(x)\}.
$
\end{definition}

Of course an MD test function is not unique, just as most powerful tests are not unique when data are discrete. We shall, however, narrow our study to MD test functions that agree with some other test function, like an LRT function so that we can focus on the effect of utilizing an MD version of an established test (ex. $\phi^*$ is an MD version of $\phi$ in the Introduction). Note that in the formal definition of MD tests that agree with another test, we say that $R(X)$ agrees with a statistic $T(X)$ if $R(x)\leq R(y)$ whenever $T(x)\geq T(y)$ for all $x,y\in \mathcal{X}$. 
\begin{definition}[MD tests, decision functions, and p-values that agree with $T$]\label{MDT}
Let $\phi_\alpha^{T}(x)$ be a size-$\alpha$ test function as in Definition \ref{T} and $\phi_\alpha^{MD}(x)$ be an MD test function as in Definition \ref{MD}.  If $T(x)\geq T(y)$ implies $R(x)\leq R(y)$ for all $x,y \in \mathcal{X}$ then we say $T$ and $R$ agree, and $\phi_\alpha^{MD-T}$ is an MD test that agrees with $\phi_\alpha^T$.  $M$-decision functions and generalized $M$-$p$-values are $\delta_{\alpha}^{M-T}(x,u) = I(u\leq \phi_{\alpha}^{MD-T})$ and $P^{M-T}(x,u) = \inf\{\alpha:\delta_\alpha^{M-T}(x,u)=1\}$, respectively.
\end{definition}
\noindent Note that if $u$ is generated then $P^{M-T}(x,u)$ is called a minimally randomized $p$-value that agrees with $T$ (MR-$T$) and written $P^{MR-T}(x,u)$.  If $u$ is specified then $P^{M-T}(x,u)$ is a minimally discrete $p$-value that agrees with $T$ (MD-$T$) and written $P^{MD-T}(x,u)$. Similar terminology is used for decision functions. 

Let us pause to consider a concrete example.  Table \ref{ExampleTable} provides an illustration of MD natural $p$-values that agree with the natural $p$-values from the usual LRT in Example \ref{BernoulliTestFunction}.  
\begin{table}[ht]
\caption{\label{ExampleTable} The first 5 columns list an $x \in \mathcal{X}$ in Example \ref{BernoulliExample}.  The following columns contain $p_0(x)$, $p_1(x)$, likelihood ratio $\lambda(x)$, a ranking $R(x)$ that agrees with $\lambda(x)$, and natural MD $p$-values based on $R(x)$ computed $P^{MD-\lambda}(x,1) = \Pr_0\{R(X)\leq R(x)\}$ and natural $p$-value based on the likelihood ratio statistic $P^\lambda(x,1) = \Pr_0\{\lambda(X)\geq \lambda(x)\}$}
\begin{tabular}{cccccccccll}
$x_1$&$x_2$&$x_3$&$x_4$&$x_5$&$p_0(x)$ & $p_1(x)$ & $\lambda(x)$ & $R(x)$ & $P^{MD}(x,1)$& $P^\lambda(x,1)$\\ \hline
1 & 1 & 1& 1& 1 & 0.03125 & 0.32768 & 10.4857 & 1 & 0.03125 & 0.03125\\
0 & 1 & 1& 1& 1 & 0.03125 & 0.08192 & 2.62144 & 2 & 0.06250 & 0.1875 \\
1 & 0 & 1& 1& 1 & 0.03125 & 0.08192 & 2.62144 & 3 & 0.09375& 0.1875\\
1 & 1 & 0& 1& 1 & 0.03125 & 0.08192 & 2.62144 & 4 & 0.125& 0.1875\\
1 & 1 & 1& 0& 1 & 0.03125 & 0.08192 & 2.62144 & 5 & 0.15625& 0.1875 \\
1 & 1 & 1& 1& 0 & 0.03125 & 0.08192 & 2.62144 & 6 & 0.1875& 0.1875\\
0 & 0 & 1& 1& 1 & 0.03125 & 0.02048 & 0.65536 & 7 & 0.21875&0.5\\
1 & 0 & 0& 1& 1 & 0.03125 & 0.02048 & 0.65536 & 8 & 0.25&0.5\\ 
\end{tabular}
\end{table}
Observe that $R$ agrees with $\lambda$ since $\lambda(x)\geq\lambda(y)$ implies $R(x)\leq R(y)$ for each $x$,$y$ listed.  We further observe that $R$ is one-to-one since each $x$ is assigned a unique ranking for the support points depicted.  Hence, the resulting $p$-value $P^{MD-\lambda}(x,1)$ is minimally discrete.  Less technically, the terminology ``minimal'' emphasizes that among all possible other natural $p$-values based on the likelihood ratio test, this $p$-value has as many support points as possible (32 in this example), and hence minimizes discreteness effect.  The p-value based on the LRT is not minimal since it only has 6 support points when 32 are possible.  

Figure \ref{TheoremPlot} displays the cumulative distribution function (CDF) for the p-values above under hypotheses $H_0:\theta = 0.5$ and $H_1:\theta = 0.8$. It is easy to see that, in this example, the usual LRT p-value is dominated by its MD counterpart: the MD-LRT natural p-value is always stochastically less than or equal to the LRT p-value but still stochastically greater than or equal to a uniform distribution under the null hypothesis. It is noteworthy that we plugged $1$ in for $u$ in $P^{MD-\lambda}(x,u)$ in this example.   

\begin{figure}[ht]
\includegraphics[width=5.0in, height = 3in]{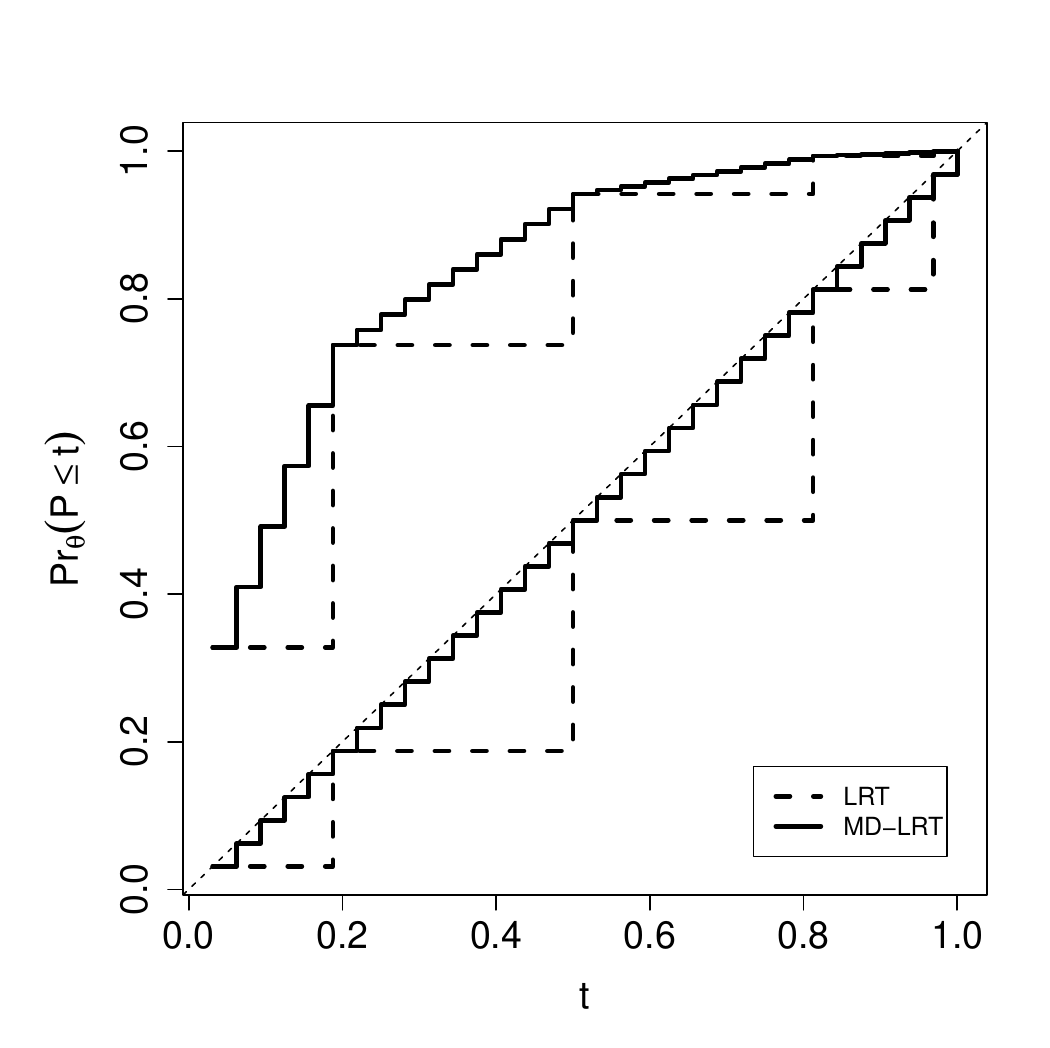}
\caption{\label{TheoremPlot} The CDF of natural p-value $\Pr_\theta\{P^{\lambda}(X,1)\leq t\}$ versus t for the LRT in Example \ref{BernoulliExample} under $H_0: \theta = 0.5$ and $H_1:\theta = 0.8$ is depicted by the dashed lines.  The CDF for its minimally discrete counterpart $\Pr_\theta\{P^{MD-\lambda}(X,1)\leq t\}$ versus t under $H_0:\theta = 0.5$ and $H_1:\theta = 0.8$ is depicted by solid lines. The dotted line represents the CDF for a uniform distribution. }
\end{figure}

We may also plug $1/2$ in for $u$ to recover an MD mid-p-value for the LRT or generate $u$ to compute a minimally randomized $p$-value $P^{MR-\lambda}(x,u)$ for the LRT. 
\begin{figure}
\includegraphics[width=5in, height = 3in]{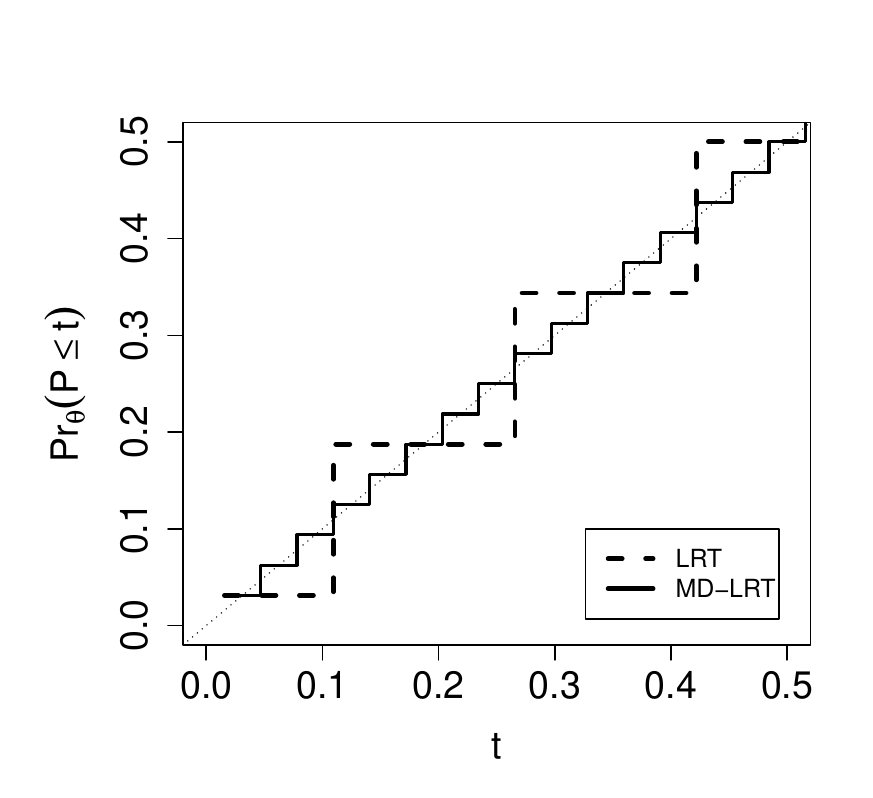}
\caption{\label{MidPplot} The CDF of the mid-p-value $\Pr_0\{P^{LRT}(X,1/2)\leq t\}$ in Example \ref{BernoulliExample} under $H_0: \theta = 0.5$ is depicted by a dashed line (zoomed in).  The CDF for its minimally discrete counterpart $\Pr_0\{P^{MD-LRT}(X,1/2)\leq t\}$ under $H_0:\theta = 0.5$ is represented by a solid line. The dotted line represents the CDF for a uniform distribution.}
\end{figure}
Figure \ref{MidPplot} depicts the distributions of the mid-p-value in Example \ref{BernoulliExample} and its MD counterpart under $H_0$. Now we see that the CDF of the MD mid-$p$-value is ``closer'' to a uniform CDF.  We also recall that both $P^\lambda(X,U)$ and $P^{MR-\lambda}(X,U)$ have uniform distributions under $H_0$ as per \eqref{valid p-value}.  The next section formally verifies that these observations are generally true.

\section{Assessment}
First let us compare MD natural p-values to their non-MD counterparts.  The Theorem below establishes that MD-natural p-values are stochastically greater than or equal to a uniform variate (i.e., valid) and stochastically less than or equal to their non-MD counterparts. This results in more powerful, yet still valid, decision rules. These validity and power results are formally stated in Claims (C1) and (C2).  The stochastic ordering claims for the corresponding $p$-value statistics are formally stated in (C3) and (C4).   

\begin{theorem}[MD-$T$ versus $T$: Natural p-values and decision functions] \label{Thm1}
Let decision functions and p-values be defined as Definitions \ref{T} and \ref{MDT} with $u=1$. The following claims are true:
\begin{enumerate}
\item[(C1)] $E_\theta[\delta_\alpha^{MD-T}(X,1)]\geq E_\theta[\delta_\alpha^T(X,1)]$  $\forall\alpha$ and $\theta\in\Theta$,\\
\item[(C2)] $E_0[\delta_\alpha^T(X,1)\leq E_0[\delta_\alpha^{MD-T}(X,1)]\leq \alpha$ $\forall \alpha$, \\
\item[(C3)] $\Pr_\theta\{P^{MD-T}(X,1)\leq t\} \geq \Pr_\theta\{P^T(X,1)\leq t\}$ $\forall t$ and $\theta\in\Theta$, \\
\item[(C4)] $\Pr_0\{P^{T}(X,1)\leq t\} \leq \Pr_0\{P^{MD-T}(X,1)\leq t\}\leq t$ $\forall t$.
\end{enumerate}
\end{theorem}

\begin{proof}
To prove (C1), if $\{\phi_{\alpha}^{MD-T}(x)=1\}\supseteq  \{\phi_{\alpha}^{T}(x)=1\}$ then $\delta_\alpha^{MD-T}(x,1)\geq \delta_\alpha^{T}(x,1)$ by construction.  Taking expectation of both sides would then prove the result.  To see that $\{\phi_{\alpha}^{MD-T}(x)=1\}\supseteq  \{\phi_{\alpha}^{T}(x)=1\}$, we assume that it is not true and derive a contradiction.  First, note that because $\mathcal{X}$ is countable and because $\phi^{MD-T}$ is minimally discrete, we can enumerate elements of $\{\phi_\alpha^{MD-T}(x) = 1\}$ by $x_{(1)}, x_{(2)}, ..., x_{(r-1)}, x_{(r)}, x_{(r+1)}, ...$ where $r = |\{\phi_\alpha^{MD-T}(x) = 1\}|$. Now, if $\{\phi_{\alpha}^{MD-T}(x)=1\}\subset  \{\phi_{\alpha}^{T}(x)=1\}$ and $T$ and $R$ agree, then $\{\phi_{\alpha}^{T}(x)=1\}$ contains $x_{(1)}, x_{(2)}, ...,x_{(r)}, x_{(r+1)}$.  But this, coupled with the fact that $\phi_\alpha^T(x)\in[0,1]$ implies $E_0[\phi_\alpha^{T}(X)]>\alpha$, which contradicts the assumption that $\phi_\alpha^T$ has size $\alpha$.   

To prove claim (C2), note the first inequality follows from (C1).  For the second inequality, observe that $$\delta_\alpha^{MD-T}(x,1) = I(1\leq \phi_\alpha^{MD-T}(x))\leq \phi_\alpha^{MD-T}(x)$$ by construction. Taking expectation gives $E_0[\delta_\alpha(X,1)]\leq E_0[\phi_\alpha(X)] = \alpha$ where the last equality is due to the supposition that $\phi_\alpha$ is a size-$\alpha$ test. 

As for claims (C3) and (C4), by Theorem 1 in \cite{Habiger2015} $I(P(x,u)\leq \alpha) = \delta_\alpha(x,u)$ for every $x\in\mathcal{X}$, $u\in[0,1]$ and $\alpha\in[0,1]$.  Taking $u=1$, we recover the results from claims (C1) and (C2), respectively.
\end{proof}

Next let us compare MR-$T$ $p$-values to randomized $p$-values based on $T$ that need not be MR.  Claim (C5) stipulates that both $p$-values are uniformly distributed under the null hypothesis.  Claim (C6) states that if $T$ is a sufficient statistic, then both $p$-value statistics are equal in distribution whether the null hypothesis is true or false. This is perhaps not surprising given that most powerful size $\alpha$ decision rules are not unique in discrete testing (recall the Introduction) and both tests are size $\alpha$ tests based on sufficient statistics. However, the MR $p$-value dominates its non-MR counterpart in that the generation of $u$ has ``minimal impact'' on $p(x,u)$.  In particular, Claim (C7) states that the variance attributable to the auxiliary $U$ is smaller for the MR-$T$ $p$-value if $T$ is sufficient. Consequently, variation from replication to replication that is attributable to generation of $U$ is minimized.  In this sense, we may say that the replicability of $p^{MR-T}(X,U)$ is greater.        
\begin{theorem}[MR-$T$ versus $T$: Randomized p-values and decision functions] \label{Thm2}
Let decision functions and p-values be defined as Definitions \ref{T} and \ref{MDT} with $U$ being uniformly distributed over the unit interval and independent of $X$. Then claim (C5) below is true. If additionally $T$ is a sufficient statistic, then claim (C6) and (C7) are true: 
\begin{enumerate}
\item[(C5)] $E_0[\delta_\alpha^{MR-T}(X,U)] = E_0[\delta_\alpha^{T}(X,U)] = \alpha$ $\forall \alpha\in[0,1]$ and hence \\
$Pr_0\{\Pr(P^{MR-T}(X,U)\leq t\} = \Pr_0\{P^{T}(X,U)\leq t\} = t$ $\forall t \in [0,1]$,\\
\item[(C6)] $E_\theta[\delta_\alpha^{MR-T}(X,U)] =  E_\theta[\delta_\alpha^T(X,U)]$ $\forall \alpha\in[0,1],\theta\in\Theta$ \\
and hence $Pr_\theta\{\Pr(P^{MR-T}(X,U)\leq t\} = \Pr_\theta\{P^{T}(X,U)\leq t\}$ $\forall t\in[0,1], \theta\in\Theta$, \\
\item[(C7)] $Var_\theta(P^{MD-T}(X,U)|X)\leq Var_\theta(P^{T}(X,U)|X)$ $\forall \theta\in\Theta$.
\end{enumerate}
\end{theorem}

\begin{proof}
For (C5), the result holds since both tests are size $\alpha$ by construction and the second result follows from Theorem 1 in \cite{Habiger2012}, which stipulates that for all $\theta\in\Theta$ and $\alpha \in[0,1]$ we have $I(P(X,U)\leq \alpha) = \delta_\alpha(X,U)$ for all $\alpha$ with probability 1 since $p$-values are defined as in Definitions \ref{T} and \ref{MDT}. 

To prove the first claim in (C6), we show that 
\begin{eqnarray*}
E_\theta[\delta_\alpha^{MD-T}(X,U)|T(X)] &=& E_0[\delta_\alpha^{MD-T}(X,U)|T(X)]\\
&=& E_0[\delta_\alpha^T(X,U)|T(X)]\\
&=& E_\theta[\delta_\alpha^T(X,U)|T(X)]
\end{eqnarray*}
and the result follows by the law of iterated expectation. Note that the first and third equalities immediately follow from the sufficiency principle since the distribution of the decision function cannot depend on $\theta$ conditionally upon $T(X)$. For the second equality, it suffices to show that $$E_0[\phi^{MD-T}_\alpha(X)|T(X)] = I(T(X)>k(\alpha)) + \gamma(\alpha)I(T(X)=k(\alpha))$$ since the right hand side is $\phi_\alpha^T(X)$. First, observe that if $T(x)>k(\alpha)$ then $R(x)<k^*(\alpha)$ since the tests agree and consequently $E_0[I(R(X)<k^*(\alpha))|T(X)>k(\alpha)] = I(T(X)>k(\alpha))$. Likewise, $T(x)<k(\alpha)$ implies $R(x)>k^*(\alpha)$.  Now, if $T(X)=k(\alpha)$, then it must be the case that $E_0[\phi_\alpha^{MD-T}(X)|T(X)=k(\alpha)] = \gamma(\alpha)$; otherwise $E_0[E_0[\phi_\alpha^{MD-T}(X)|T(X)]] \neq \alpha$. This completes the proof of the first claim in (C6).  The second claim follows from Theorem 2.3 in \cite{Habiger2011}, which stipulates $I(P(X,U)\leq \alpha) = \delta_\alpha(X,U)$ with probability 1 for $p$-values defined as in Definitions \ref{T} and \ref{MDT}. 

To prove claim (C7), observe for a fixed $x\in \mathcal{X}$, $0\leq \Pr_\theta\{R(X)=R(x)\}\leq \Pr_\theta\{T(X)=T(x)\}$ because $$\textrm{Pr}_\theta\{T(X)=T(x)\} = \sum_{x\in A} \textrm{Pr}_\theta\{R(X)=R(x)\}$$ where $A = [T(X)=T(x)]$.  Hence, \begin{eqnarray*}
Var_\theta(P^{MD-T}(x,U)|X=x) &=& Var(U)[\textrm{Pr}_\theta\{R(X)=R(x)\}]^2\\
&\leq& Var(U)[\textrm{Pr}_\theta\{T(X)=T(x)\}^2.
\end{eqnarray*}
\end{proof}

Finally, let us compare mid-$p$-values based on $T$ to their MD counterparts.  Recall mid-p-values are bounded by a uniform variate in the convex order under $H_0$, i.e. $E_0[h(P^T(X,1/2))]\leq E_0[h(U)]$ for all convex $h$. We can verify that sharper inequalities are available, i.e. $E_0[h(P^T(X,1/2))]\leq E_0[h(P^{MD-T}(X,1/2))]\leq E_0[h(U)]\leq E_0[h(U)]$.  Graphically, the result is understood in Figure \ref{MidPplot}, where we see the cumulative distribution of MD-$T$ mid-$p$-value statistic is sometimes closer to the 45 degree line and never further. The result is formally stated in claim (C9).  Claim (C8) provides the analogous convex ordering results for $T$ and $MD-T$ test functions.  Again, the claims require $T$ be a sufficient statistic, though relaxing the assumption may be possible.  
   
\begin{theorem}[MD-$T$ versus $T$: Mid p-values and test functions] \label{Thm3}
Let $\phi_\alpha^T(X)$, $\phi_\alpha^{MD-T}(X)$, $P^{T}(X,1/2)$ and $P^{MD-T}(X,1/2)$ be test functions and mid-p-values defined as in Definitions \ref{T} and  \ref{MDT}. If $T$ is sufficient then the following claims are true for any convex $h$:
\begin{enumerate}
\item[(C8)] $E_0[h(\phi_\alpha^T(X))]\leq E_0[h(\phi_\alpha^{MD-T}(X))]$ $\forall \alpha \in [0,1]$.  \\
\item[(C9)] $E_0[h(P^T(X,1/2))]\leq E_0[h(P^{MD-T}(X,1/2)]\leq E_0[h(U)].$
\end{enumerate}
\end{theorem}
\begin{proof}
For Claim (C8), we use the Martingale representation of convex ordering in Theorem 3.A.4 in \cite{Shaked2007}.  From this Theorem, it suffices to show $E[\phi^{MD-T}_\alpha(X)|\phi^T_\alpha(X)] = \phi^T_\alpha(X)$.   
The idea is akin to the proof of Claim (C6), but we cannot directly use the sufficiency principle since $\phi^T_\alpha(X)$ is not a sufficient statistic.  However, observe $\phi_\alpha^{T}(x) = 1 $ implies $ [T(x)>k(\alpha)]$ which, because $T$ and $R$ agree, implies $R(x)<k^*(\alpha)$.  Thus, $\phi_\alpha^{MD-T}(x) = 1$ if $\phi_\alpha^{T}(x) =1$.  Similarly, $\phi^T_\alpha(x) = 0$ implies $\phi^{MD-T}_\alpha(x)=0$. Finally, $\phi^{T}_\alpha(x)\in (0,1)$ implies $\phi^T_\alpha(x) = \gamma(\alpha)$ and hence $T(x) = \gamma(\alpha)$ by construction.  The constraint that $E_0[E_0\{\phi^{MD}(X)\}|\phi^T(X)] = \alpha$ ensures that $E_0[\phi^{MD}(X)|T(X) = \gamma(\alpha)] = \gamma(\alpha)$.  Thus,  
$E_0[\phi^{MD-T}_\alpha(X)|\phi_\alpha^T(X)] = \phi_\alpha^{T}(X).$ 

For Claim (C9), using a convex order representation in \cite{Shaked2007}, the proof amounts to verifying that a) the area under the curve and to the left of $s$ for the MD mid-p-value CDF is greater than or equal to the area under the curve and to the left of $s$ for the usual mid-p-value CDF under $H_0$ and that b) both p-values have expectation 1/2 under $H_0$. The former claim is observable in Figure \ref{MidPplot} since the areas of interest are sums of rectangles.  However, the formal proof is notationally complex and therefore relegated to the Appendix.
\end{proof}

\section{Applications and Remarks}

A brief discussion regarding the potential impact of results here is warranted. First consider multiple hypothesis testing. A multiple hypothesis testing (MTP) procedure tests null hypotheses $H_{01},H_{02},..., H_{0M}$ with $p$-values $P_1, P_2, ..., P_M$, which are often assumed to be uniformly distributed under the null hypotheses. Formally, they are defined $\delta_i = I(P_i\leq t)$ for $i = 1, 2, ..., M$ and some $t\in[0,1]$.  For example, a Bonferroni procedure chooses $t = \alpha/M$ while the \cite{Benjamini1995} procedure for False Discovery Rate control chooses $\hat{t} = \sup\{s:\frac{s}{\max\{\sum_i I(P_i\leq s),1\}}\leq \alpha\}$ (with the supremum of the empty set defined to be 0), which can be shown to converge to some constant $t$ as $M\rightarrow \infty$ with probability 1 under weak dependence (cf. \cite{Genovese2002,Storey2004, Habiger2015}).  While procedures are often studied under the assumption that $\Pr_0(P_i\leq t) = t$,  this assumption can be relaxed to $\Pr_0(P_i\leq t)\geq t$. Claim (C4) in Theorem \ref{Thm1} ensures that, if opting for natural $p$-values, the MD-natural $p$-values are still valid since they satisfy the later relaxed assumption.  Claims (C1) and (C3) give that they result in more power.  If opting for randomized $p$-values then both MR $p$-values and their non-minimal counterparts are uniformly distributed under the null hypotheses. If randomized $p$-values are functions of sufficient statistics, then we should not anticipate any more or less power if opting for their MR counterparts (claim (C6)), as the Neyman Pearson Lemma suggests. However, claim (C7) stipulates that fewer discoveries will rely on the generation of $u$ if utilizing the MR $p$-values. To see this, consider Example \ref{BernoulliTestFunction} and consider implementing a level $\alpha = 0.05$ test.  In this example, the decision to reject or retain $H_0$ will depend upon $U$ with probability $\Pr\{T(X) = 4\}$, which is 5 times more likely than its MR counterpart that only utilizes $u$ for one $x\in [T(x)=4]$.    

Theorem 3 has implications for meta analysis with mid-p-values and mid-p-value calibration. For example \cite{Rubin-Delanchy2018} showed that Fisher's method for combining $p$-values is asymptotically conservative when $P_1, P_2, ..., P_M$ are mid-p-values.  That is,
$\Pr_0\left\{-2\sum_{i=1}^n\log(P_i)\geq t_{\alpha,n}\right\}\leq \alpha$ for large enough $n$ where $t_{\alpha,n}$ is a critical value for a chi-squared distribution with $2n$ degrees of freedom and $\Pr_0$ represents a probability computed under the global null (assuming all $H_{0m}$s are true).  This is attributable to the fact $-2\log(x)$ is convex, with convex ordering closed under convolution \citep{Shaked2007}. Hence $-2\sum _{i=1}^n \log(P_i)\leq_{cx}-2\sum _{i=1}^n\log(U_i)$ under the global null when $p$-values are independent.  For $P_i^{MD}$ minimally discrete versions of $P_i$, this inequality is sharpened by claim (C9). Several other Hoeffding-type inequalities were utilized in \cite{Rubin-Delanchy2018} to verify validity of meta analysis methods with mid-p-values, and stand to be sharpened as well. Likewise, \cite{Wang2024} showed that $\Pr(\tilde{P}\leq \alpha)\leq \textrm{e}\times \alpha$
where $\tilde{P} = \exp(\sum_{i=1}^M w_k \log P_i)$ for $P_i$ mid-p-values.  This result also lead to valid meta anlaysis methods, which stand to be improved by claim (C9) via the sharpening of the inequality. Results here may also be useful in the conversion of mid-p-values to E-values \citep{Vovk2021} for downstream analysis.  For example, \cite{Wang2024} demonstrated that many admissible methods for converting p-values to $E$-values \citep{Vovk2021} (ex. $E_i = kP_i^{k-1}$ for $k\in(0,1)$ or $E_i = P_i^{-1/2}-1$) are valid if $P_i$ is a mid-p-value.  Note that such results also rely on convex ordering of mid-p-values.  

As mentioned in the Introduction and throughout this manuscript, most powerful level $\alpha$ test functions are not unique.  In fact, even MD test functions and their corresponding MD and MR p-values are not unique! This means that the selection of specific MD or MR p-values in practice may be based on auxiliary information, prior knowledge, secondary aims, technical variation across experimental units, intuition, subject matter knowledge, or may even be arbitrary. For example, in Table \ref{ExampleTable}, auxiliary information for ranking $x_1$, $x_2$, ..., $x_5$ based on any of the above criteria may be available, which when coupled with the likelihood ratio, would be sufficient for defining $R(x)$.  In other settings, ``less discrete'' and ``less randomized'' p-values may be constructed based on statistical considerations.  For example, one could justify taking $R(x) = 2$ for both $x = (1,1,1,1,0)$ and $x = (0,1,1,1,1)$ and taking $R(x)=3$ for other $x\in [T(x)=4]$ since these two realizations would provide more evidence against an assumption that $X_i$ are independent.  While a detailed exploration of ranking criteria and other types of less discrete or less randomized p-values is beyond our scope, this paper clearly establishes (in Claims (C1) - (C9)) a need to more carefully consider how p-values are constructed, especially when they are inputs for downstream analysis. 
   
\section*{Appendix: Proof of Claim (C9)}

First note that the last inequality in (C9) is established as in equation \eqref{sub-uniform}.  To establish the first inequality, we utilize Theorem 3.A.1 in \cite{Shaked2007}, which states that for $X$ and $Y$ random variables with $E[X] = E[Y]$, $E[h(X)]\leq E[h(Y)]$ for all convex $h$ iff  
$$
\int_{-\infty}^sF(u)du\leq \int_{-\infty}^sG(u)du
$$
for all $s$, where $F$ and $G$ are the cumulative distribution functions for $X$ and $Y$ respectively. It will be illuminating to refer to Figure \ref{MidPplot} in what follows.  

First, let us adopt some simplifying notation and verify that both mid-p-values have mean 1/2. Denote the CDF of $T(X)$ under $H_0$ by $F_T$ and corresponding pmf by $f_T$, and likewise denote the CDF of $R$ by $F_R$ and pmf by $f_R$.  Further denote the mid-p-values in equation \eqref{MD} and \eqref{MDT} by $P^T = P^T(X,1/2)$ and $P^R = P^{MD}(X,1/2)$. Finally, we write $\Pr$ for $\Pr_0$ for brevity.  Observe $\Pr\left\{P^T = \bar{F}(t) + 0.5f_T(t)\right\} = f_T(t)$ for $\bar{F}(t) = 1-F(t)$ and $\Pr\left\{P^R = F_R(r-1) + 0.5 f_R(r)\right\} = f_R(r)$.  Writing $F_R(r-1) + 0.5 f_R(r) = \frac{1}{2}[F_R(r-1) + F_R(r)]$ and $f_R(r) = F_R(r) - F_R(r-1)$ we see that 
$$E[P^R] = \sum_r [F_R(r-1) + 0.5 f_R(r)] f_R(r) = \frac{1}{2}\sum_r [F_R(r)^2 - F_R(r-1)^2] = 1/2$$
where the last equality is due to $F_R(0)=0$, $\sum_{r=1}^N[F(r)^2 - F(r-1)^2] = F(N)^2$ and $\lim_{N\rightarrow\infty}F(N)=1$.  Similar arguments can be used to verify $E[P^T] = 1/2$ under $H_0$. 

To see that $\int_0^s\Pr(P^T\leq \alpha)d\alpha\leq \int_0^s\Pr(P^R\leq \alpha)d\alpha$, we shall show that the integrals are equal over the region $[\Pr(T>k),\Pr(T\geq k))$ for all $k$ and with ``$\leq$'' occurring over $[\Pr(T>k),s)$ if $[\Pr(T>k)< s<\Pr(T\geq k))$. To see that the integrals are equal partition
$
[\Pr(T>k), \Pr(T\geq k))\equiv A^-(k)\bigcup A^+(k)
$
into two intervals of width $0.5\Pr_0(T=k)$ so that $A^-(k)  = [\Pr(T>k), \Pr(T>k) + 1/2 \Pr(T=k))$ and  
$A^+(k) = [\Pr(T>k) + 1/2 \Pr(T=k), \Pr(T\geq k))$.  
Observe $P^T = \Pr(T>k)$ on $A^-(k)$ and $P^T = \Pr(T\geq k)$ on $A^+(k)$ so that
$\int_{A^-(k)\cup A^+(k)} \Pr(P^T\leq \alpha)d\alpha = 0.5\Pr(T>k) + 0.5 \Pr(T\geq k)$.   
If $\{R(x): T(x)=k\}$ has cardinality 1 then $P^T=P^R$ on $A^-(k)\cup A^+(k)$ since $T$ and $R$ agree and hence the integrals are equal.  Now suppose $T(X) = k$ for two values of $R$. Then because $T(X)$ is sufficient and because $T(X)$ and $R(X)$ agree, we can find a $k^*$ satisfying $\Pr\left\{T(X)>k\right\} = \Pr\left\{R(X)<k^*\right\}$ and $\Pr\left\{T(X)\geq k^*\right\} = \Pr\left\{R(X)\leq k^*+1\right\}$ and can partition $[\Pr(T>k),\Pr(T\geq k))$ into 4 equal width intervals, each with width $1/2f_R(R=k^*) = 1/2f_R(k^*+1) = 1/4\Pr(T = k)$. See, for example, Figure \ref{MidPplot}.  Formally,  
$A^-(k)\cup A^+(k) = B^-(k^*)\cup B^+(k^*) \cup B^-(k^*+1) \cup B^+(k^*+1)$ for
\begin{eqnarray*}
B^-(k^*) &=& [F_R(k^*-1), F_R(k^*-1)+ 0.5f_R(k^*))\\
B^+(k^*) &=& [F_R(k^*-1)+ 0.5f_R(k^*), F_R(k^*)) \\
B^-(k^*+1) &=& [F_R(k^*), F_R(k^*)+ 0.5f_R(R=k^*+1)) \\
B^+(k^*+1) &=& [F_R(k^*) + 0.5 f_R(k^*+1), F_R(k^*+1))
\end{eqnarray*}
and $\Pr(P^R\leq \alpha)$ is $F_R(k^*-1)$, $F_R(k^*)$, $F_R(k^*), F_R(k^*+1)$ on $B^-(k^*),..., B^+(k^*+1)$, respectively.  Thus, the integral is 4 equal width rectangles with average height $1/2 [\Pr(T>k) + \Pr(T\geq k)]$.  More formally, computing the integral and then adding and subtracting $f_r(k^*)$ to $F_R(k^*)$ to get $F_R(k^*+1)$ and $F_R(k^*-1)$ gives
\begin{eqnarray*}
\int_{A^-(k)\cup A^+(k)}&&\hspace{-.3in}\Pr(P^R\leq \alpha)d\alpha\\
&=& 1/2 f_R(k^*)[F_R(k^*-1)+F_R(k^*) + F_R(k^*) + F_R(k^*+1)] \\
&=& 1/2 f_R(k^*)[2F_R(k^*-1) + 2F_R(k^*+1)]\\
&=& f_R(k^*)[F_R(k^*-1) + F_R(k^*+1)]\\
&=& 1/2 f_T(k)[F_R(k^*-1) + F_R(k^*+1)]\\
&=& 1/2 f_T(k)[\Pr(T>k) + \Pr(T\geq k)] \\
&=&\int_{A^-(k)\cup A^+(k)}\Pr(P^T\leq \alpha)d\alpha.
\end{eqnarray*}
Now if $T(x)=k$ for $j\geq 2$ values of $R$ then $\Pr[\Pr(T>k),\Pr(T\geq k))$ can be partitioned into 2$j$ intervals with equal width $1/2f_R(k^*)$, with $\Pr(P^R\leq \alpha)$ having average height $1/2 [\Pr(T>k) + \Pr(T\geq k)]$. Hence, we always have
$\int_{A^-(k)\cup A^+(k)}\Pr(P^R\leq \alpha)d\alpha = \int_{A^-(k)\cup A^+(k)}\Pr(P^T\leq \alpha)d\alpha$. 

Now for the region $[\Pr(T>k),s)$, we observe that there are $j$ rectangles in $\int_{A^-(k)}[\Pr(P^R\leq \alpha) - \Pr(P^T\leq \alpha)]d\alpha$ with width $1/2f_R(k^*)$ and height $\Pr(R<j) - \Pr(T>k)$ for $j=k^*, k^*+1,..$. Further, $\int_{A^+(k)}[\Pr(P^T\leq \alpha) - \Pr(P^R\leq \alpha)]d\alpha$ is composed of an identical collection of rectangles.  Therefore, for $s\in A^-(k)$
\begin{eqnarray*}
\int_{[\Pr(T>k),s))}[\Pr(P^R\leq \alpha) -\Pr(P^T\leq \alpha)]d\alpha\geq 0.
\end{eqnarray*}
and for $s\in A^+(k)$
\begin{eqnarray*}
&&\int_{[\Pr(T>k),s))}[\Pr(P^R\leq \alpha) -\Pr(P^T\leq \alpha)]d\alpha \\
&&=\int_{A^-(k)}[\Pr(P^R\leq \alpha) -\Pr(P^T\leq \alpha)]d\alpha \\
&&\hspace{.3in} - \int_{[\Pr(T>k)+1/2f_T(k), s)}[\Pr(P^T\leq \alpha) -\Pr(P^R\leq \alpha)]d\alpha \\
&&\geq 
\int_{A^-(k)}[\Pr(P^R\leq \alpha) -\Pr(P^T\leq \alpha)]d\alpha \\
&&\hspace{.3in} - \int_{A^+(k)}[\Pr(P^T\leq \alpha) -\Pr(P^R\leq \alpha)]d\alpha \\
&& = 0
\end{eqnarray*}
Thus, for any $s$ we can compute (for some $k$)
\begin{eqnarray*}
\int_0^s\Pr(P^R\leq \alpha)d\alpha &-& \int_0^s\Pr(P^T\leq \alpha)d\alpha\\
&&\hspace{-.5in}= \int_{\Pr(T>k)}^s\Pr(P^R\leq \alpha)d\alpha - \int_{\Pr(T>k)}^s\Pr(P^T\leq \alpha)d\alpha \geq 0.
\end{eqnarray*}
That is, $\int_0^s\Pr(P^R\leq \alpha)d\alpha\geq \int_0^s\Pr(P^T\leq \alpha)d\alpha$ and $E_0[P^R]=E_0[P^T]$.  This completes the proof. 

\bibliographystyle{chicago}
\bibliography{multipletesting}
\end{document}